\documentclass{article}
\usepackage{amsmath, amsthm, amsfonts}
\usepackage[utf8]{inputenc}
\usepackage{graphicx}
\usepackage{natbib}
\usepackage{comment}
\usepackage{color}
\usepackage{caption}
\usepackage[a4paper,margin=1.25in,footskip=0.25in]{geometry}
\usepackage{authblk}

\usepackage{xr}
\usepackage{multirow}

\newtheorem{assumption}{Assumption}
\newtheorem{theorem}{Theorem}
\newtheorem{definition}{Definition}

\newcommand{\bias}{\text{Bias}}
\newcommand{\cov}{\text{Cov}}
\newcommand{\E}{\mathbb{E}}

\newcommand{\spp}{\text{sp}}

\newcommand{\indep}{\perp \!\!\! \perp}
\newcommand{\yitrue}{Y_i^{(t)}}
\newcommand{\yirep}{Y_i^{(r)}}

\DeclareMathOperator*{\argmin}{argmin}

\title{Robust Designs for Prospective Randomized Trials Surveying Sensitive Topics}

\author{Evan T. R. Rosenman}
\affil{Harvard University Data Science Initiative}

\author{Rina Friedberg}
\affil{LinkedIn Data Science and Applied Research}

\author{Mike Baiocchi}
\affil{Stanford University School of Medicine}

\begin{document}

\maketitle

\abstract{We consider the problem of designing a prospective randomized trial in which the outcome data will be self-reported, and will involve sensitive topics. Our interest is in misreporting behavior, and how respondents' tendency to under- or overreport a binary outcome might affect the power of the experiment. We model the problem by assuming each individual in our study is a member of one ``reporting class": a truth-teller, underreporter, overreporter, or false-teller. We show that the joint distribution of reporting classes and ``response classes" (characterizing individuals' response to the treatment) will exactly define the bias and variance of the causal estimate in our experiment. Then, we propose a novel procedure for deriving sample sizes under the worst-case power corresponding to a given level of misreporting. Our problem is motivated by prior experience implementing a randomized controlled trial of sexual violence prevention program among adolescent girls in Nairobi, Kenya.}

\section{Introduction}

In a wide swath of social science, political science, and public health research, outcome data is drawn from self-reported survey responses  \citep{cimpian2020mischievous, imai2010causal, zaller1992simple}. Accurate responses are paramount to the integrity of such research -- and, as a result, there is a vast literature discussing best practices for surveying sensitive topics, including specially designed survey instruments \citep[see e.g.][]{droitcour2001three}. Nonetheless, the literature contains myriad examples of surveys whose conclusions were skewed by subpopulations of misreporters, especially when the respondents were adolescents \citep{cimpian2020mischievous}.

Statistical corrections have been developed for a variety of problems involving mismeasured data \citep{carroll2006measurement}. Early work in this area focused on error that is stochastically independent of outcomes and covariates, and how it might propagate through the linear model. More modern methods have moved beyond the classical assumption that measurement error is uncorrelated with outcomes, and have sought to adjust model-based procedures to mitigate bias, as well as properly account for uncertainty due to the error \citep{hausman2001mismeasured}. In the setting of causal inference, there has been a substantial research interest in the issue of treatment compliance \citep{frangakis2002principal, balke1997bounds}, and the related issue of measurement error in treatment assignment variables \citep{imai2010causal, lewbel2007estimation}. The problem of measurement error in outcome variables has been comparatively understudied. 

Our focus in this article is on mismeasurement in binary outcomes for causal inference. Our motivation arises from prior experience designing and analyzing a randomized trial of a sexual violence prevention program targeting adolescent girls in Kenya \citep{rosenman2020empirical}. In that study, girls self-reported key outcomes via survey responses. Inaccurate reports -- also known as ``reporting biases" -- represented a threat to causal inference. 

Here, we operate in a prospective setting, supposing that decisions are being made about how to formulate and conduct a study, including the number of units to recruit in order to achieve a desired power level. 
We posit that -- for a given survey instrument -- the tendency to misreport a binary outcome is a fixed characteristic of the individual. We thus define each individual in our population as a member of a single ``reporting class": either a  truth-teller, underreporter, overreporter, or false-teller. 

We separately invoke the notion of a ``response class" \citep{hernan2010causal} to characterize how each individual responds to the treatment. Under our assumptions, the joint distribution of response classes and reporting classes fully determines the bias and variance in estimating the average treatment effect in a randomized trial. This notion yields direct insights into how to better design a prospective trial when under- and overreporting behavior is suspected. 

The remainder of this paper proceeds as follows. Section \ref{sec:relatedWork} reviews prior literature in the areas of measurement error and response bias. Section \ref{sec:definitions} introduces our definition of reporting classes and provides motivating examples. Section \ref{sec:results} introduces our main theoretical results, including a novel procedure for computing the required sample size to achieve a desired power under key assumptions about the population frequency of misreporters. Section \ref{sec:workedExample} works through an example based on real-life data from Kenya. Section \ref{sec:specialCase} discusses a useful extension when researchers have acess to additional information via a pilot study. Section \ref{sec:conclusion} concludes.

\section{Related Work} \label{sec:relatedWork}

Our work builds on themes from the psychometrics literature, wherein researchers design survey measurements and analysis practices to more accurately measure sensitive or self-reported data (see for example, \citet*{pickles}). Some survey respondents are hesitant to report sensitive data, or could be harmed if their individual responses were observed. \citet*{droitcour2001three} address this problem with the three-card method, which protects the privacy of the individual while still allowing for unbiased population estimates. \citet*{doi:10.1177/2332858419888892} address bias in estimates due not to concerns about privacy for sensitive information, but due to mischievous responders. Their work addresses students who falsely self-report as lesbian, gay, bisexual, or questioning, and correspondingly skew estimates in studies that aim to understand and serve this community. Instead of changing the survey estimates, they leverage statistical techniques, such as boosted regressions with outlier removal, to address bias in the survey responses.

Regarding measurement error in treatments, \citet*{flegel} show how errors in exposure measurement can lead to errors in classification, wherein misclassification probabilities for disease predictions differed by disease status. They focus on non-differential measurement error, meaning the errors do not differ by study arm (as opposed to differential errors, where exposure or reporting measurement error can differ in the treatment and control groups). As a solution, \citet{10.1093/aje/kwaa208} suggest finding an unexposed reference group and estimating a disease score that estimates unintended exposure effects. This can confer conditional independence between the potential outcomes under no exposure, and the measured covariates. In another approach, \citet{coleAJE} address measurement error in exposures for longitudinal observational studies, suggesting either multiple imputation or Bayesian methods.
 
 Our work focuses on measurement error in outcomes. \citet*{doi:10.1080/03610926.2014.887105} give a general parametric model for misclassified binary outcomes that vary with covariates, showing how to find consistent estimates for model parameters and individual misclassification probabilities. \citet{10.1093/aje/kwz133} propose a sensitivity analysis to determine how strong differential measurement error in either the treatment or the outcome would need to be, in order to explain the entire result. Their method compares the true effect of the exposure with the association between the exposure and the mismeasured outcome, on scales that depend on risk ratios of the true outcome probabilities conditioned on the exposure, and the mismeasured probabilities, respectively. 
 
Tools from the causal inference literature have been extremely useful in recent literature studying differential measurement error. \citet*{imai2010} derive bounds of the average treatment effect under differential measurement error in a binary treatment variable, and introduce a sensitivity analysis to help the researcher understand how robust the estimate is to such measurement errors. \citet*{colecausal} also leverage insights from causal inference to discuss measurement error; in fact, our work is most similar in spirit to their research. The authors use potential outcomes to formalize a discussion of ``missing data" due to measurement error in exposures. We build on this by introducing a framework for reporting behavior and deriving the statistical properties of treatment effect estimates under different reporting behaviors. 

In this way, our work also builds on \citet*{angristcausalIV}, in which the authors use instrumental variables analysis as a replacement for intention-to-treat methods. Our construction of reporting classes is directly inspired by compliance classes, which describe how individuals respond to an optional treatment assignment. As an example, consider individuals receiving a voting incentive -- any voter can either receive the incentive or not (randomly), and can either vote or ignore the incentive. This maps to four compliance classes. There are never-takers, who would never take the treatment (here, who would never vote); always-takers, who would always take it no matter what; compliers, who take the treatment to which they are assigned; and defiers, who do the opposite. 

\section{Definitions and Examples} \label{sec:definitions}
In this section, we work through an example, introduce terminology, and point to a handful of constraints that may be easily missed.

\subsection{Reporting classes}\label{sec:reportClass}

Consider a randomized trial in which students are randomly assigned to receive either a violence prevention program (``intervention'') or a training unrelated to violence prevention (``control''). The intervention's goal is to reduce the students experience of violence, so the outcome is a binary: ``yes, I experienced violence in the prior 12 months'' $(=1)$ or ``no, I did not experience violence in the prior 12 months'' $(=0)$. We will operate in the potential outcomes framework \cite{rubin1974estimating}, meaning each individual $i$ in the trial $i$ has two associated outcomes: $Y_i(1)$, the realized outcome if given the intervention; and $Y_i(0)$, the realized outcome if given the control. For simplicity, we will stipulate that the stable unit treatment value assumption (SUTVA) holds, meaning that the outcome for each individual in the trial is unaffected by the assignment of treatments to the other units \citep{imbens2015causal}. 
Lastly, for ease of explication, we first discuss this study assuming there are only two types of reporting classes of students: truth-tellers and underreporters. 

A truth-teller does what we usually assume: only if they experienced violence will they report having experienced violence. An underreporter, however, reports not having experienced violence under any circumstance -- e.g., whether they actually have experienced violence or not. 

We can hypothesize causes for this kind of reporting behavior, in which the true response is ``yes" but the reported response is ``no." Suppose there are three causes of underreporting in this setting: (i) students fear negative consequences to themselves if they report a violent event, (ii) students face feelings of shame for having been a victim of violence, or (iii) students are mistaken or confused about what constitutes an experience of violence. 



Careful consideration of such hypotheses may help researchers to modify and design a prospective study to mitigate misreporting bias in the design phase. For example, (i) to address concerns of negative consequences of the report, the researchers may choose to use a technique like ``randomized response,'' which offers a higher degree of anonymity in response, (ii) to reduce feelings of shame, researchers may choose to word the the question to avoid triggers of shame, and (iii) to clarify what constitutes ``experiencing violence," researchers may use detailed scenarios of what does and does not constitute violence; use specific physical and verbal acts; or use common slang terms which are more familiar to students. 



\subsection{Response classes and reporting classes}\label{sec:RRclasses}

In the prior section, we considered one type of grouping for individuals in our study: reporting classes. We now turn to a second type of grouping, \emph{response} classes. 

Response classes are sets of participants with the same potential outcomes. For a useful prior discussion of this concept, see \cite{hernan2010causal}. We use the following naming for the response classes: 
\begin{itemize}
    \item \emph{decrease} for $(Y_i(0)=1,Y_i(1)=0)$,
    \item \emph{increase} for $(Y_i(0)=0,Y_i(1)=1)$,
    \item \emph{never} for $(Y_i(0)=0,Y_i(1)=0)$, and
    \item \emph{always} for $(Y_i(0)=1,Y_i(1)=1)$.
\end{itemize}
Note that our naming convention for response classes departs from the one used by Hern\'an and Robin. We are trying to avoid valences for the outcome directions (e.g., instead of ``helped,'' we use ``increase''). 

A key subtlety: for a specific intervention and control, being in a certain response class may preclude an individual from being in a particular reporting class. For example, note that students do not have the ability to underreport if they would never experience violence regardless of intervention level; they would always report ``no.'' Hence, we say that individuals in the ``never" response class cannot be underreporters. A more formal discussion of this idea can be found in Section \ref{sec:simpAssumptions}. 


As described technically in Section \ref{sec:biasResults}, the joint distribution of response classes and reporting classes determines the potential for bias. To understand this point, consider how response classes and reporting classes might jointly vary in the violence prevention example. First consider the ``decrease'' response class $(Y_i(0) = 1, Y_i(1) = 0)$. There may be students inside of this group who feel that they can modify their chances of experiencing violence, and also may feel shame if they are unable to prevent the violence. Contrast that group with the ``always'' group $(Y_i(0) = Y_i(0) = 1)$. The ``always'' group may not feel violence can be modified and therefore feels more comfortable reporting their experiences truthfully. In this scenario, we might imagine underreporting would exist more among those responsive to the treatment. This scenario would be particularly challenging for unbiased estimation because it means that there will be more underreporting among the decrease response class than among the always response class. This will bias estimation of the treatment effect toward $0$, much more so than if underreporting behavior were equally probable among the decrease and always groups.

There are at least two features of the study protocol that
determine reporting classes. First, there is how the measurement is obtained. This includes aspects of the measurement tool, such as question wording and response type (e.g., multiple choice, open-response), as well as the context in which it is asked (e.g., computer-based survey versus one-on-one interviews). Second, there are the specifics of the intervention and control. For example, if researchers change the control level to be more intensive (e.g., by using a ``standard of care'' rather than an ``attention control'') then some participants may shift their response class from ``increase'' to ``always'' and therefore shift their reporting class from overreporters to truth-tellers. 

We take the reporting class to be fixed after the intervention levels and the measurement procedure are fixed. Our primary motivation is describing the challenges of the measurement itself, helping to anticipate and improve how participants react to being measured. Defining the reporting class in this way precludes an important type of measurement errors that arise differentially by arm. For example, we do not allow for the possibility that the violence prevention program itself induces individuals to underreport experiences of violence more than they might have otherwise. Such ``demand effects" are important, but outside the scope of this manuscript.

\subsection{A remark about false-tellers}\label{sec:remarkFalse}

We have not offered an example of false-tellers. Can false-tellers exist? Yes, but like defiers in the literature on instrumental variables, we suspect they are not as common as the other types. The motivation to avoid telling the truth -- in fact, being willing to report either outcome-level, just not the true outcome-level -- is a more complicated dynamic than the other motivations discussed above. We are aware of anecdotal examples of quarrelsome study participants who might function as false-tellers. 

\section{Theoretical Results} \label{sec:results}

In this section, we formalize the insights described in the prior section, and then introduce our main theoretical results. 

\subsection{Preliminaries}
\subsubsection{Model, Assumptions, and Notation}

We consider a super-population of $i = 1, \dots, N_{\spp}$ units, of whom $2n << N_{\spp}$ units will be sampled for a completely randomized experiment. Once the experimental units are drawn, exactly $n$ of the units are selected via simple random sample to receive the treatment, while the remaining $n$ units are assigned to the control condition. With each unit we associate a pair of fixed, binary potential outcomes $(Y_i(0), Y_i(1)) \in \{0, 1\}^2$, representing the unit's true response in the presence or absence of treatment. Our target of estimation will be the super-population treatment effect,
\[ \tau_{\spp} = \frac{1}{N_{\spp}} \sum_{i = 1}^{N_{\spp}} Y_i(1) - Y_i(0) \,. \] 

We will denote by $R_i \in \{0, 1\}$ the indicator of being sampled into the experiment
. Among units chosen to be in the experiment, we define a second indicator, $W_i \in \{0, 1\}$, such that
\[ W_i = \left\{ \begin{array}{cc} 1 & \text{if unit $i$ is treated} \\ 0 & \text{if unit $i$ is control} \end{array} \right. \] 
The two processes are supposed independent, such that $R_i \indep W_j$ for all $i, j$. We denote as $\E_R(\cdot)$ expectation with respect to the $R_i$ and $\E_{W}(\cdot)$ expectation with respect to the $W_i$. For any estimator $\phi$, we define 
\[ \E(\phi) = E_R\left(E_W \left(\phi \mid \{R_i \}_{i = 1}^N\right)\right)\,. \] 
Note also that $E_R(R_i) = 2n/N_{sp}$ and $E_W(W_i) = n_t/n$. 

Our interest is in the setting of under- and overreporting. Hence, we associate with each unit two additional fixed binary constants: 
\[ U_i = \left\{ \begin{array}{cc} 1 & \text{if $i$ is an underreporter} \\ 0 & \text{otherwise} \end{array} \right.,  \hspace{5mm}  O_i = \left\{ \begin{array}{cc} 1 & \text{if $i$ is an overreporter} \\ 0 & \text{otherwise} \end{array} \right. \] 

Per the discussion in Section \ref{sec:remarkFalse}, we make the following assumption. 
\begin{assumption}\label{assumption:noDefiers}
Every unit in the super-population is either an underreporter, an overreporter, or a truth-teller. There are no ``false-tellers" -- those who always report the opposite outcome of the truth. 
\end{assumption}
Under Assumption \ref{assumption:noDefiers}, we see both that $U_i O_i = 0$ for all units and that $U_i = O_i = 0$ implies a unit is a truth-teller. We will sometimes refer collectively to the group of underreporters and overreporters as ``misreporters.''

We define two versions of the outcome: $\yitrue$, the true outcome, and $\yirep$, the reported outcome. $\yitrue$ is defined in the standard way: 
\[ \yitrue = W_i Y_i(1) + (1-W_i) Y_i(0) \,.\] 
The reported outcome has a more complex structure: 
\[ \yirep = (1-U_i)(1-O_i) \left( W_i Y_i(1) + (1-W_i) Y_i(0)\right) + O_i\,. \] 
By inspection, one can see that the above definition yields the expected reporting behavior. Regardless of the value of $\yitrue$, the reported outcome is 0 if $U_i = 1$ and is 1 if $O_i = 1$. Only if $U_i = O_i = 0$ do we have $\yirep = \yitrue$. 

The target of inference is $\tau_{\spp}$, the super-population causal effect. We define $\bar U_{\spp}$ and $\bar O_{\spp}$ as the super-population averages of $U_i$ and $O_i$. 

\subsubsection{Response Classes}

We define four response classes -- ``Decrease", ``Increase", ``Never", and ``Always" -- reflecting the individual's potential outcomes with or without the treatment. This terminology is slightly more generic than the one often used in epidemiology: ``helped," ``hurt," ``immune," and ``doomed" \citep{hernan2010causal}. Our reasoning is that we would like to consider cases where $Y_i = 1$ is a beneficial, rather than deleterious, outcome. Our terminology does not assign a value judgment to the state of the outcome. 

We define binary variables $D_i, I_i, N_i, A_i \in \{0, 1\}$ reflecting whether each individual $i$ falls into the Decrease, Increase, Never, or Always classes. Every individual must belong to exactly one of the four response classes, so we have that \[ D_i + I_i + N_i + A_i = 1 \text{ for all $i$.} \] 

Using these definitions, we can summarize our super-population via the following table. We give the response class on the left and the reporting class on the top. Boxes marked with an ``x" do not exist due to assumptions. In the remaining boxes, the given symbol denotes the population fraction of units that fall into the given category. We drop the ``sp" subscript for simplicity.

\begin{table}[h]
\centering
\begin{tabular}{l|ll|llll|l}
                & \textbf{Y(0)} & \textbf{Y(1)} & \textbf{\begin{tabular}[c]{@{}l@{}}Truth-\\ teller\end{tabular}} & \textbf{\begin{tabular}[c]{@{}l@{}}Over-\\ reporter\end{tabular}} & \textbf{\begin{tabular}[c]{@{}l@{}}Under-\\ reporter\end{tabular}} & \textbf{\begin{tabular}[c]{@{}l@{}}False-\\ teller\end{tabular}} &  \\ \hline
\textbf{Decrease} & 1             & 0             &   $\overline{T D}$                                                               &  $\overline{O D}$                                                                 &     $\overline{U D}$                                                               & x                                                                & $\overline{D}$ \\
\textbf{Increase}   & 0             & 1             &  $\overline{T I}$                                                                &  $\overline{O I}$                                                                 & $\overline{U I}$                                                                   & x                                                                & $\overline{I}$ \\
\textbf{Never} & 0             & 0             & $\overline{T N}$                                                                 & $\overline{O N}$                                                                   & $\overline{U N} $                                                                 & x                                                                & $\overline{N}$ \\
\textbf{Always} & 1             & 1             &  $\overline{T A}$                                                                & $\overline{O A} $                                                                & $\overline{U A}$                                                                   & x                                                                & $\overline{A}$ \\ \cline{1-8}
                &              &              &  $\overline{T}$                                & $\overline{O}$                                                & $\overline{U}$                                                &                                                              & 
\end{tabular}
\end{table}

\subsubsection{Further Simplifying Assumptions}\label{sec:simpAssumptions}

Eagle-eyed readers will observe that our current definitions of under- and overreporters engenders a certain level of ambiguity. To make this subtlety clearer, we focus specifically on underreporters and draw out the following table of possible orientations of our binary indicators and outcomes. 

\begin{table}[h]
\centering
\begin{tabular}{|l|l|l||l|l||l|l|}
\hline
\textbf{Response Class} & $\boldsymbol{Y_i(0)}$ & $\boldsymbol{Y_i(1)}$ & $\boldsymbol{W_i}$ & $\boldsymbol{U_i}$ & $\boldsymbol{\yitrue}$ & $\boldsymbol{\yirep}$ \\ \hline
Decrease & 1    & 0    & 1 & 1 & 0               & 0      \\ \hline
Decrease & 1    & 0    & 0 & 1 & 1               & 0      \\ \hline
Increase & 0    & 1    & 1 & 1 & 1               & 0      \\ \hline
Increase & 0    & 1    & 0 & 1 & 0               & 0      \\ \hline
Never & 0    & 0    & 1 & 1 & 0               & 0      \\ \hline
Never & 0    & 0    & 0 & 1 & 0               & 0      \\ \hline
Always & 1    & 1    & 1 & 1 & 1               & 0      \\ \hline
Always & 1    & 1    & 0 & 1 & 1               & 0      \\ \hline
\end{tabular}
\end{table}

In many of the rows, we observe $\yitrue = \yirep$ -- that is, even if a subject is an underreporter, her reported outcome is the same as her true outcome, because her true outcome is 0 under a given treatment assignment. For the Decrease, Increase, and Always response class, there exists at least one treatment assignment such that being an underreporter changes the reported outcome from the true outcome. But for the Never response class, the reported outcome and true outcomes are always 0, regardless of treatment assignment. Hence, the designation as an underreporter would be a pure abstraction for an individual in the Never response class, since she could not exhibit underreporting behavior under any randomization of the treatment assignment. 

To resolve this ambiguity, we offer the following definition. 

\begin{definition}\label{def:noAmbig}
A unit can be an underreporter or an overreporter if and only if there exists some treatment assignment under which $\yirep \neq \yitrue$. 
\end{definition}

Under this definition, an individual in the Never response class cannot be an underreporter, while an individual in the Always response class cannot be an overreporter. All other configurations are possible, because there exists at least one value of $W_i$ (and possibly both values) such that the reported outcome will differ from the true outcome. 

Under Definition \ref{def:noAmbig}, we revise the proportions table to exclude two more possible combinations. This yields the following table: 

\begin{table}[h]
\centering
\begin{tabular}{l|ll|llll|l}
                & \textbf{Y(0)} & \textbf{Y(1)} & \textbf{\begin{tabular}[c]{@{}l@{}}Truth-\\ teller\end{tabular}} & \textbf{\begin{tabular}[c]{@{}l@{}}Over-\\ reporter\end{tabular}} & \textbf{\begin{tabular}[c]{@{}l@{}}Under-\\ reporter\end{tabular}} & \textbf{\begin{tabular}[c]{@{}l@{}}False-\\ teller\end{tabular}} &  \\ \hline
\textbf{Decrease} & 1             & 0             &   $\overline{T D}$                                                               &  $\overline{O D}$                                                                 &     $\overline{U D}$                                                               & x                                                                & $\overline{D}$ \\
\textbf{Increase}   & 0             & 1             &  $\overline{T I}$                                                                &  $\overline{O I}$                                                                 & $\overline{U I}$                                                                   & x                                                                & $\overline{I}$ \\
\textbf{Never} & 0             & 0             & $\overline{T N}$                                                                 & $\overline{O N}$                                                                   & x                                                                  & x                                                                & $\overline{N}$ \\
\textbf{Always} & 1             & 1             &  $\overline{T A}$                                                                & x                                                                & $\overline{U A}$                                                                   & x                                                                & $\overline{A}$ \\ \cline{1-8}
                &              &              &  $\overline{T}$                                & $\overline{O}$                                                & $\overline{U}$                                                &                                                              & 
\end{tabular}
\end{table}

\subsection{Bias Results}\label{sec:biasResults}
We consider the bias of the finite sample difference-in-means estimator, defined as 
\[ \hat \tau = \frac{1}{n} \sum_{i = 1}^{N_{\spp}} Y_i W_iR_i - \frac{1}{n} \sum_{i = 1}^{N_{\spp}} Y_i(1-W_i)R_i \,.\] 

\begin{theorem}[Bias of Difference-in-Means Estimator]\label{thm:bias}
Define $\tau_i = Y_i(1) - Y_i(0)$ for $i = 1, 2, \dots, N_{sp}$ such that 
\[ \tau_{sp} = \frac{1}{N_{sp}} \sum_{i = 1}^{N_{sp}} \tau_i \,.\] 
The bias of the difference-in-means estimator in estimating $\tau_{sp}$ is given by 
\begin{align*}
    \bias(\hat \tau) &= -\left(\left(\bar U_{\spp} + \bar O_{\spp}\right) \tau_{\spp} + \cov_{\spp}\left(U, \tau \right) + \cov_{\spp}\left(O, \tau \right)\right) \\
&= - \overline{UI} + \overline{UD} - \overline{OI}  + \overline{OD}
\end{align*} 
where $\cov_{\spp}\left(U, \tau \right)$ is the super-population covariance between $U_i$ and $\tau_i$ and $\cov_{\spp}\left(O, \tau \right)$ is the super-population covariance between $O_i$ and $\tau_i$. 
\end{theorem}
\begin{proof}
See the Appendix, Section \ref{biasProof}. 
\end{proof}

Theorem \ref{thm:bias} shows us that there are several factors at play in determining the bias due to reporting behavior. If, in the super-population, we have $U_i, O_i \indep \tau_i$, then our estimate will be shrunk toward 0 by a multiplicative factor: 
\[ \frac{\E \left( \hat \tau \right)}{\tau_{\spp}} = 1 - \bar U - \bar O \hspace{5mm} \text{ under independence. } \] 
If independence does not hold, the bias may point in either direction, depending on the sign of the covariance terms. For example, if $\tau_{\spp}$ is positive, the bias may also be positive if those in the Decrease response class are disproportionately likely to be under- or overreporters.

\subsection{Power Under Independence}

These results can be extended to a power analysis. We consider Neyman's null hypothesis versus a one-directional alternative, 
\begin{align*}
H_0 &: \tau_{\spp} = 0\\
H_1 &: \tau_{\spp} < 0 \,. 
\end{align*}

As is clear from the bias results, it is \emph{not} true that non-differential measurement error must reduce detection power. However, power  does strictly fall with misreporter incidence when the treatment effect is independent of under- and overreporter status. This condition is equivalent to the stating that under- and overreporters are no more or less prevalent in any response class than any other (excepting the responses classes in which there are no under- or overreporters by assumption). 

This intuition is formalized in Theorem \ref{thm:indepPower}. 

\begin{theorem}\label{thm:indepPower}
Suppose $\tau_{\spp} < 0$ and, in the super-population, we have  $U_i$ \\$\indep Y_i(0), Y_i(1) \mid N_i = 0$ and $O_i \indep Y_i(0), Y_i(1) \mid A_i = 0$. Then the detection power is a strictly decreasing function of $\bar U_{\spp}$ and $\bar O_{\spp}$. 
\end{theorem}
\begin{proof}
See the Appendix, Section \ref{proof:indepPower}. 
\end{proof}

In the absence of the independence condition, detection powers can plausibly increase or decrease with misreporter incidence. In the next section, we introduce the concept of worst-case power under a sensitivity model.

\subsection{Sensitivity Model}

We have introduced two perspectives for considering the relationship between the treatment effect and misreporting behavior. The first perspective, in which we consider covariance quantities like $\cov_{\spp}\left(U, \tau \right)$, and $\cov_{\spp}\left(O, \tau \right)$, has been useful for proving theorems. The latter perspective, in which we consider joint incidence proportions like $\overline{UI}$ and $\overline{OD}$, is useful for quantifying  possible deviations from independence. In this section, we use this feature to define a sensitivity model. 

We posit the existence of some value $\Gamma \geq 1$ such that 
\begin{equation}\label{eq:gammaBounds}
    \begin{split}
    \frac{1}{\Gamma} \leq \frac{\overline{UI}}{\overline{U} \cdot \overline{I}} \cdot (1 - \overline{N}) \leq \Gamma, \hspace{3mm} \frac{1}{\Gamma} \leq \frac{\overline{UD}}{\overline{U} \cdot \overline{D}} \cdot (1 - \overline{N}) \leq \Gamma, \hspace{3mm}  \frac{1}{\Gamma} \leq \frac{\overline{UA}}{\overline{U} \cdot \overline{A}} \cdot (1 - \overline{N}) \leq \Gamma, \\
    \frac{1}{\Gamma} \leq \frac{\overline{OI}}{\overline{O} \cdot \overline{I}} \cdot (1 - \overline{A}) \leq \Gamma, \hspace{3mm} \frac{1}{\Gamma} \leq \frac{\overline{OD}}{\overline{O} \cdot \overline{D}} \cdot (1 - \overline{A}) \leq \Gamma, \hspace{3mm} \frac{1}{\Gamma} \leq \frac{\overline{ON}}{\overline{O} \cdot \overline{N}} \cdot (1 - \overline{A}) \leq \Gamma. 
    \end{split}
    \end{equation}
The expressions in (\ref{eq:gammaBounds}) need to be unpacked. We refer to the ratios within each inequality as the ``response class ratios." Recall that, by definition, there are no underreporters among the Never class and no overreporters among the Always class. Hence $\overline{UN}$ and $\overline{OA}$ are both equal to 0 -- and there are three response classes in which undereporter incidence can be greater than zero, and three response classes in which overreporter incidence can be greater than zero. 

The given inequalities bound how much variation in underreporter and overreporter incidence can occur across the relevant response classes. Concretely, if $\Gamma = 4/3$, this implies that underreporters and overreporters can be no more than 33\% more but no less than 25\% less frequent within any of the three relevant response classes than they are on average across the three relevant response classes.


Note that the response class ratios are naturally bounded, e.g. for $\overline{UI}/(\overline{U} \cdot \overline{I})$, we have
\[ \frac{\overline{U} + \overline{I} - 1}{\overline{U} \cdot \overline{I}} \leq \frac{ \overline{UI}}{\overline{U} \cdot \overline{I}}  \leq \min\left(\frac{1}{\overline{I}}, \frac{1}{\overline{U}} \right) \] 
with analogous conditions for each of the other  response class ratios. We will typically assume that $\Gamma$ is sufficiently close to 1 such that the interval bounds in  (\ref{eq:gammaBounds}) are more restrictive than these bounds. Lastly, observe that 
\[ \overline{UI} + \overline{UD} + \overline{UA} = \overline{U} \hspace{3mm} \text{and} \hspace{3mm} \overline{OI} + \overline{OD} + \overline{ON} = \overline{O} \] 
must also be satisfied.

For ease of expressing the sample size computation as an optimization problem, we will denote each of the response class ratios as its own variable, with the $\Gamma$ bounds functioning as an upper and lower bound. For example, we define
\[ \delta_{UI} = \frac{\overline{UI}}{\overline{U} \cdot \overline{I}} \hspace{3mm} \text{ and hence } \hspace{3mm} \frac{1}{\Gamma} \leq \delta_{UI} \leq \Gamma \,,\] 
with analogous definitions for the other seven response class ratios. Key quantities can be easily expressed as a function of the response class ratios. For example, we can now write the bias as 
\begin{align*}
    \bias(\hat \tau) &= - \overline{U}\cdot\overline{I} \delta_{UI} + \overline{U}\cdot\overline{D}\delta_{UD} - \overline{O}\cdot\overline{I}\delta_{OI}  + \overline{O}\cdot\overline{D}\delta_{OD}.
\end{align*} 

\subsection{Power and Sample Size Calculations}

Under a given choice of $\Gamma$, we can now consider the question of sample-size calculation. Let us suppose an analyst is designing an experiment for a treatment intended to have a negative average treatment effect. The analyst can tolerate the probability of a Type II error no greater than $\beta \in [0, 1]$ (so her desired power is $1 - \beta$). The analyst also has a desired probability of a Type I error no greater than $\alpha \in [0, 1]$. 

To proceed, the analyst must propose prospective values of $\overline{I}, \overline{D}, \overline{N}, \overline{A}, \overline{U}, \overline{O}$, as well as a value for $\Gamma$. We collect the former variables and the response class ratios into vectors, defined as 
\begin{align*}
\boldsymbol \pi &= \left( \overline{I}, \overline{D}, \overline{N}, \overline{A}, \overline{U}, \overline{O} \right)^T \\
\boldsymbol \delta &= \left(\delta_{UI}, \delta_{UD}, \delta_{UA}, \delta_{OI}, \delta_{OD}, \delta_{ON}\right)^T.
\end{align*}
We can define the expected mean outcome among the treated and control units as functions of these values:
\begin{align*}
\mu_1\left(\boldsymbol \pi, \boldsymbol \delta \right) &= \overline{I} + \overline{A} - \overline{U} \left(\overline{I}\delta_{UI} + \overline{A} \delta_{UA} \right) + \overline{O} \left(\overline{D} \delta_{OD} + \overline{N} \delta_{ON}\right) , \\
\mu_0\left(\boldsymbol \pi, \boldsymbol \delta \right) &= \overline{D} + \overline{A} - \overline{U} \left(\overline{D}\delta_{UD} + \overline{A} \delta_{UA} \right) + \overline{O} \left(\overline{I} \delta_{OI} + \overline{N} \delta_{ON}\right).
\end{align*}
Next, we can express the standardized expectation of the difference-in-means estimator as 
\[ T(\boldsymbol \pi, \boldsymbol \delta, n) = n \cdot  \frac{\mu_1\left(\boldsymbol \pi, \boldsymbol \delta \right) - \mu_0\left(\boldsymbol \pi, \boldsymbol \delta \right) }{\mu_1\left(\boldsymbol \pi, \boldsymbol \delta \right) \cdot \left(1 - \mu_1\left(\boldsymbol \pi, \boldsymbol \delta \right) \right) + \mu_0\left(\boldsymbol \pi, \boldsymbol \delta \right) \cdot \left(1 - \mu_0\left(\boldsymbol \pi, \boldsymbol \delta \right) \right) }. \] 

Now, we can directly pose our problem as an optimization problem. In particular, we want to solve for the minimum treatment and control sample size $n$ such that we still achieve the desired power level under all possible configurations of the response class ratios. This is directly encoded in Optimization Problem \ref{mainOptProb}:

\begin{equation}\label{mainOptProb} 
\begin{aligned}
\argmin_n \max_{\boldsymbol \delta} &\hspace{5mm} T(\boldsymbol \pi, \boldsymbol \delta, n) \\
\text{subject to} & \hspace{5mm} \Phi\left( \Phi^{-1}(\alpha) -  T(\boldsymbol \pi, \boldsymbol \delta, n) \right) \geq 1 - \beta,\\
&\hspace{5mm}  \overline{U}(\overline{I}\delta_{UI} +  \overline{D}\delta_{UD} + \overline{A}\delta_{UA}) = \overline{U},\\ 
&\hspace{5mm}  \overline{O} (\overline{I}\delta_{OI} +  \overline{D}\delta_{OD} + \overline{N}\delta_{ON}) = \overline{O},\\ &\hspace{5mm}
\frac{1}{\Gamma} \leq \mathcal{H} \boldsymbol{\delta}  \leq \Gamma\,,
\end{aligned}
\end{equation}
where $\Phi(\cdot)$ is the CDF of a standard normal, and 
\[ \mathcal{H} = \text{diag} \left( 1 - \overline{N}, 1 - \overline{N}, 1 - \overline{N}, 1 - \overline{A}, 1 - \overline{A}, 1 - \overline{A} \right). \] 
Optimization Problem \ref{mainOptProb} is a quadratic fractional program and can be solved efficiently using Dinkelbach's method \citep{dinkelbach1967nonlinear, phillips1998quadratic}. For more details, see the discussion in the Appendix, Section \ref{app:solveOpt}. 

\subsection{Simulation Results}

We now consider how these computations work in practice. Recall that the analyst must provide estimates of the response class frequencies, so we begin with the assumption that $\overline{I} = 0.1, \overline{D} = 0.2, \overline{N} = 0.35,$ and $\overline{A} = 0.35$. This implies that the average causal effect is a 10\% reduction in incidence of the outcome. 

Suppose we want to achieve a power level of 80\% at a 95\% detection threshold. We use a one-sided hypothesis test. Were we to assume the population contained no underreporters and no overreporters in the population, we would find that a sample size of 612 units -- half in treatment, and half in control -- would be sufficient. As we will see, the required sample sizes dramatically increase once we incorporate the possibility of misreporting, especially if the misreporting incidence can vary by response class. 

In Figure \ref{fig:powerPlots}, we consider the cases in which either undereporting or overreporting is suspected, but not both. We show the required sample size for the desired power level and detection threshold. The incidence of underreporters and overreporters is varied on the $x$ axis, while we consider a variety of possible values of $\Gamma$. Recall that these are \emph{worst-case} sample sizes, assuming the incidence of under- or overreporters is as adversarial as possible under the $\Gamma$ constraints. 

As we can see, when $\Gamma$ is small, the required sample sizes grows slowly with misreporter frequency. Even if 20\% of respondents were to be underreporters or overreporters, the worst-case required sample size under $\Gamma = 1$ would be 1,162 units, or about twice the required sample size with no misreporting. However, as $\Gamma$ increases, the required sample size can grow very dramatically with misreporter frequency. The required sample size is more sensitive to overreporters than to underreporters, owing to the fact that the estimated causal effect is negative. With 20\% underreporters and $\Gamma = 1.5$, the required sample size grows to about 7,000 units; with 20\% overreporters and $\Gamma = 1.5$, the required sample size grows to over 20,000 units. 

\begin{figure}
    \centering
    \includegraphics[width = \textwidth]{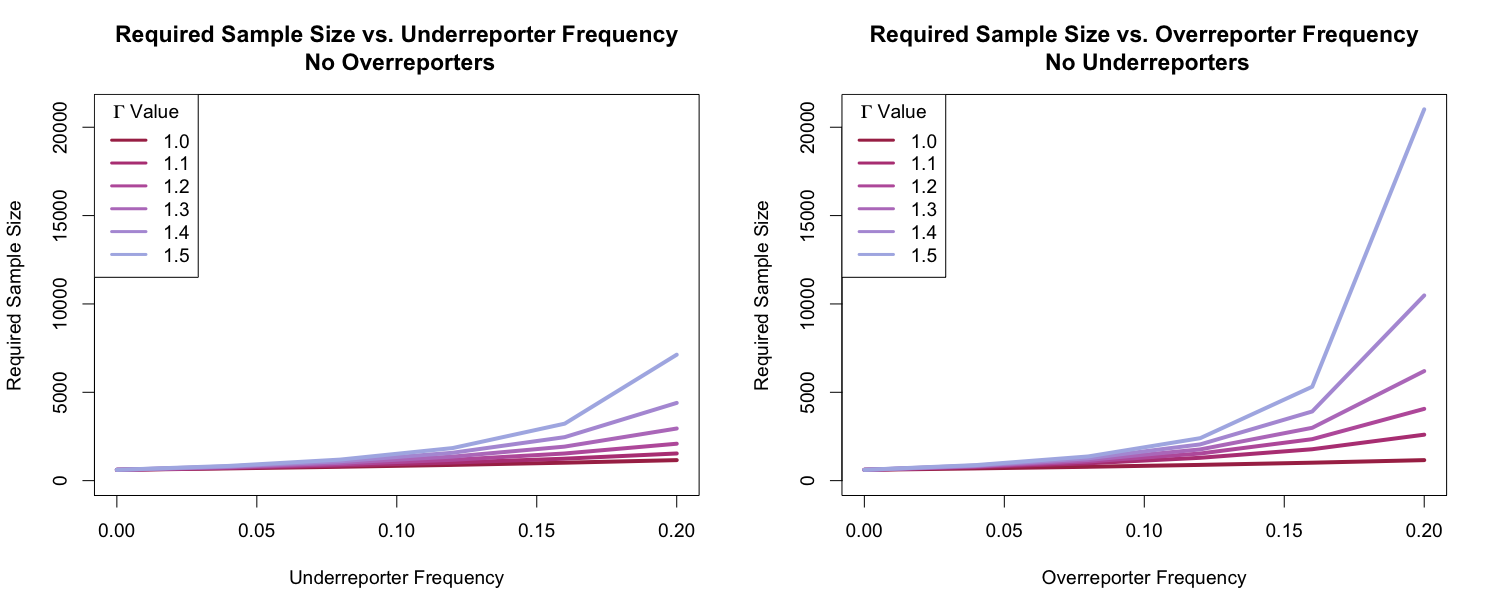}
    \caption{Required sample size for 80\% power and a detection threshold of 95\%, varying levels of $\overline U$ (left panel) and $\overline O$ (right panel), for different values of sensitivity parameter $\Gamma$. 
    \label{fig:powerPlots}}
\end{figure}

The same computations can be repeated when both undereporters and overreporters are suspected in the population. In Figure \ref{fig:powerPlotsBoth}, we repeat the process but vary the level of ``misreporters" under the assumption that half of misreporters are overreporters and half are underreporters. Hence, a misreporter frequency of 0.10 corresponds to $\overline U = 0.05$ and $\overline O = 0.05$. As we can see, the pattern is largely the same as in Figure \ref{fig:powerPlots}, but the rate of growth in required sample size at each value of $\Gamma$ is somewhere between the underreporters-only case and the overreporters-only case. At 20\% misrporters and $\Gamma = 1.5$, the worst-case sample size grows to over 12,000 units. 

\begin{figure}
    \centering
    \includegraphics[width = 0.6\textwidth]{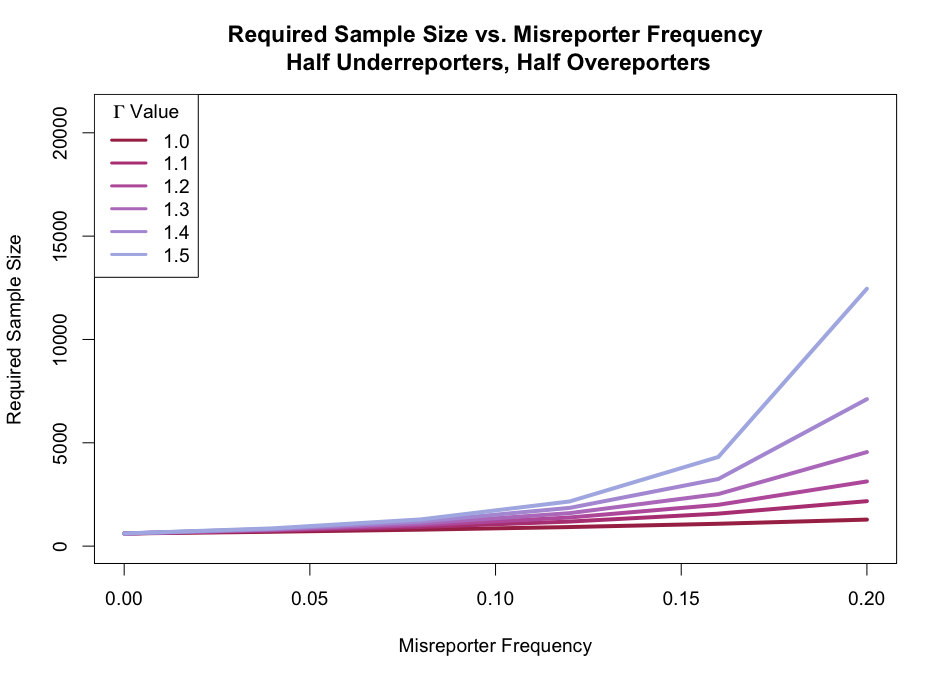}
    \caption{Required sample size for 80\% power and a detection threshold of 95\%, varying incidence of misreporters when half of misreporters are assumed to be underreporters and half are assumed to be overreporters. Lines represent different values of sensitivity parameter $\Gamma$. \label{fig:powerPlotsBoth}}
\end{figure}

\section{Worked Example}\label{sec:workedExample}

We demonstrate how researchers designing a study can account for concerns about measurement error in outcomes, using one of our prior studies. Consider a randomized trial evaluating a sexual violence prevention program for female students in sixth grade classrooms. From earlier research \citep{baiocchi2019prevalence}, we expect the annual rate of sexual violence in this population to be 7\%, and are concerned about underreporting in outcomes. 

A standard power calculation might proceed as follows. We have two populations to compare, and will observe two sample proportions of students reporting sexual violence in the prior year: $p_T$ and $p_C$. For simplicity, in this toy example we will not account for cluster structure. Suppose we use a one-sided hypothesis test and our goal is 80\% power. We assume the intervention reduces sexual violence by 50\% annually, so that $\mu = p_T - p_C = (0.07 * 0.50) - 0.07 = -0.035$. This leads to an estimate of the necessary sample size of 994 individuals. The full calculation can be found in the Appendix, Section \ref{app:powerCalc}. 

\begin{figure}
\centering
\includegraphics[width=0.7\textwidth]{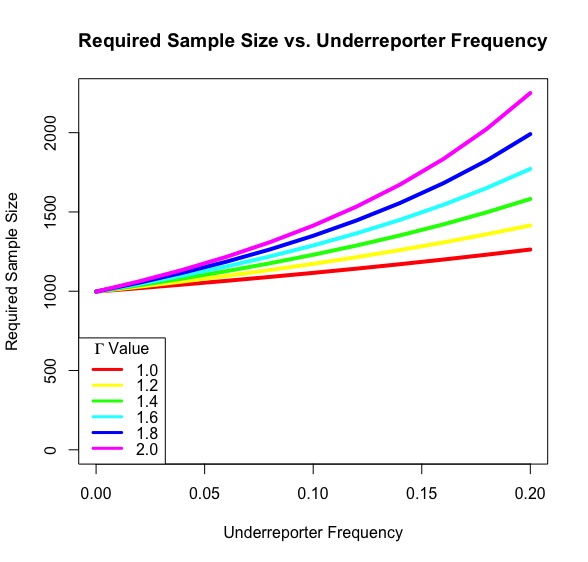}
\caption{Sample size by underreporter frequency for the classroom example.}\label{fig-worked}
\end{figure}

Let us now, however, account for underreporting. Figure \ref{fig-worked} shows the required sample size by underreporter frequency for a set of possible $\Gamma$ values. Assuming no underreporting, we remain at a suggested sample size of 1000, as in the standard power calculation case. But should we want to account for underreporting, we have a guide for the required sample size given our estimated underreporter frequency and desired sensitivity.  

\section{Special case: estimating the joint distribution of response and reporting classes}\label{sec:specialCase}

Suppose we could obtain gold standard measurements of the observed and reported outcome for an unbiased pilot sample from the population of interest. That is, for a sample of potential participants,  we could see both $(\yitrue, \yirep)$. Such data might be available by asking participants their responses first with the proposed measurement instrument (e.g., a survey), and then validating responses with a measurement that is more labor-intensive, expensive, but with less/no error (a ``gold-standard'' measure). As we will demonstrate in this section, such data can be quite useful for the purposes of improving the experimental design. This procedure can also be performed after a study has been completed, in which case the data can be used to modify measures of uncertainty and the point estimate as well as being used to provide information for planning future studies.

We use the same notation as the prior section. Each unit in the gold standard sample will receive either the treated or control condition. We suppose our sample contains $n$ treated units and $n$ control units. If $i$ is a treated unit -- then we have access to $(\yitrue(1), \yirep(1))$ representing the true and reported potential outcomes. More explicitly, these values are given by
\[ \yitrue(1) = W_i \yitrue, \hspace{5mm} \text{and} \hspace{5mm} \yirep(1) = W_i \yirep\,. \] 
Analogously, if $i$ is a control unit, then we have access to $(\yitrue(0), \yirep(0))$, where 
\[ \yitrue(0) = (1 - W_i) \yitrue, \hspace{5mm} \text{and} \hspace{5mm} \yirep(0) = (1 - W_i) \yirep\,. \] 

Now, we consider the super-population from which the data is sampled. Observe that we can relate the super-population frequencies of the true and reported outcomes to the joint distribution of our response and reporting classes defined in Section \ref{sec:simpAssumptions}. These frequencies will be estimates of population frequencies. In Table \ref{tab:yi1} and Table \ref{tab:yi0}, we show how each of these population frequencies relates to the joint response-reporting classes. Note that we are imposing Assumption \ref{assumption:noDefiers}, under which we are assume there are no false-tellers. 

\begin{table}[h]
\centering
\begin{tabular}{|c|c|c|c|}
\hline
                                                  & \hspace{10mm} & \multicolumn{2}{c|}{$Y_i^{(r)}(1)$} \\ \hline
                                                  &   & \hspace{20mm} 1\hspace{20mm} &\hspace{20mm}  0 \hspace{20mm}                        \\ \hline
\multirow{2}{*}{$Y_i^{(t)}(1)$} & 1 &      $\overline{TI} + \overline{TA} + \overline{OI}$                     &    $\overline{UI} + \overline{UA}$                       \\ \cline{2-4} 
                                                  & 0 &  $\overline{OD} + \overline{ON}$                          &          $\overline{TD} + \overline{TN} + \overline{UD}$                 \\ \hline
\end{tabular}
\caption{\label{tab:yi1}Joint reporting and response classes for treated potential outcomes.}
\end{table}

\begin{table}[h]
\centering
\begin{tabular}{|c|c|c|c|}
\hline
                                                  & \hspace{10mm}  & \multicolumn{2}{c|}{$Y_i^{(r)}(0)$} \\ \hline
                                                  &   & \hspace{20mm} 1 \hspace{20mm}                         &\hspace{20mm} 0 \hspace{20mm}                         \\ \hline
\multirow{2}{*}{$Y_i^{(t)}(0)$} & 1 &      $\overline{TD} + \overline{TA} + \overline{OD}$                     &    $\overline{UD} + \overline{UA}$                       \\ \cline{2-4} 
                                                  & 0 &  $\overline{OI} + \overline{ON}$                          &          $\overline{TI} + \overline{TN} + \overline{UI}$                 \\ \hline
\end{tabular}
\caption{\label{tab:yi0}Joint reporting and response classes for untreated potential outcomes.}
\end{table}

Crucially, under our assumptions, every cell in these tables can be estimated unbiasedly. For example, we have
\[ E \left( \frac{1}{n_t} \sum_{i: W_i = 1} I(\yitrue(1) = 1, \yirep(1) = 1) \right) = \overline{TI} + \overline{TA} + \overline{OI} \,,\] 
and 
\[ E \left( \frac{1}{n_c} \sum_{i: W_i = 0} I(\yitrue(0) = 1, \yirep(0) = 0) \right) = \overline{UD} + \overline{UA} \,.\]
Referring to the results of Theorem \ref{thm:bias}, we observe that the estimator 
\begin{align*}
\hat B &= \frac{1}{n_c} \sum_{i:W_i = 0} I(\yitrue(0) = 1, \yirep(0) = 0)  - \frac{1}{n_t} \sum_{i:W_i = 1} I(\yitrue(1) = 1, \yirep(0) = 0) \\
& + \frac{1}{n_t} \sum_{i:W_i = 1} I(\yitrue(1) = 0, \yirep(1) = 1) - \frac{1}{n_c} \sum_{i:W_i = 0} I(\yitrue(0) = 0, \yirep(1) = 1)
\end{align*}
satisfies 
\[ E(\hat B) = - \overline{UI} + \overline{UD} - \overline{OI} + \overline{OD} = Bias(\hat \tau)\,. \] 
In other words, with no further assumptions, we can obtain an estimator for the amount of bias in the usual difference-in-means estimator when using the proposed measurement instrument in an experiment. This gives very useful insight into the potential pitfalls of the experiment. If this bias is intolerably large, the experimenters may seek to change to a more accurate instrument. 

Other sample proportions may also be informative. For example, to obtain an estimate of the power, one would need an estimate of $\overline{\tilde Y_i^{(r)}(1)}$ and $\overline{\tilde Y_i^{(r)}(0)}$, the average reported potential outcomes under each condition. These can be estimated from the pilot sample by simply taking the empirical means.

Finally, under an additional assumption, one could infer the entire joint distribution of the response and reporting classes. In particular, suppose we assume that there are no individuals in the Decrease class -- i.e., the treatment can only yield an increase in the outcome of interest, relative to the control condition. Then $\bar D = 0$. By inspection of Tables \ref{tab:yi1} and \ref{tab:yi0}, one can see that all the nonzero proportions -- $\overline{TI}, \overline{TN}, \overline{TA}, \overline{OI}, \overline{ON}, \overline{UI},$ and $\overline{UA}$ -- can be estimated unbiasedly from the pilot sample. Estimating the proportions of the classes usefully decomposes the bias term into groups of individuals who may have more recognizable reasons for their reporting class. Say after running one of these ``gold-standard'' checks we see that the bias is largely arising from overreporters/increasers rather than underreporters/always then we have a better sense of how to alter our data collection protocol by considering the dynamics leading to increasers overreporting.

\section{Concluding Remarks}\label{sec:conclusion}

The formal treatment of mismeasurement and reporting bias is a somewhat recent development in the causal inference literature, especially when considering errors that may not be independent of potential outcomes \citep{imai2010causal}. This problem poses a real threat to valid causal inference -- especially when considering randomized trials in which the outcomes include intimate or sensitive topics. The problem can be partially allayed through careful survey design and other methods to ensure participant privacy. But there are many practical settings in which all participants cannot reasonably be expected to disclose data about an outcome of interest. 

In this paper, we have considered the problem from the perspective of a researcher designing a prospective randomized trial where the outcome of interest is binary. Our contributions are twofold. First, we have sought to rigorously define a set of ``reporting classes" 
to characterize individual reporting behavior. We demonstrate that the joint distribution of these reporting classes -- along with the ``response classes'' reflecting how individuals respond to treatment \citep{hernan2010causal} -- will determine exactly how much bias and variance will be induced in our causal estimate. Second, we have proposed a method for practitioners to adequately power their analyses in the presence of misreporting, assuming a worst-case deviation from independence of the reporting and response classes. 

Future opportunities in this area are myriad. We have focused on a model of misreporting that treats the response class as a fixed characteristic of an individual, given the measurement instrument. Recent work has considered the case in which misreporting behavior is differential \citep{10.1093/aje/kwz133}, meaning individuals may misreport only when they receive the treatment or only when they receive the control. This introduces additional complexities into the task of determining adequate sample size, but would represent an important step toward a realistic model of misreporting behavior. 

\newpage
\section*{Acknowledgments and Funding}
We thank Luke Miratrix, Sophie Litschwartz, Kristen Hunter, and Julia Simard for their useful comments and feedback. We are also grateful to Google and to the Marjorie Lozoff Fund, provided by the Michelle R. Clayman Institute for Gender Research, for financial support. 
This research was conducted while Rina Friedberg was a student at Stanford, and is not affiliated with LinkedIn.

\newpage

\bibliographystyle{apalike}
\bibliography{biblio}

\appendix
\section{Example: Sexual Debut Delay Program}

Here, we provide another motivating example for understanding misreporting behavior in randomized controlled trials. 

There is a complicated literature on sexual debut in adolescents and its connection to both short-term and long-term health and behaviors. Suppose we are asked to study an existing program intended to delay sexual debut among adolescents. In this scenario, we consider the concurrent presence of over- and underreporters, as well as truth-tellers. Some study participants may underreport debut to reduce their chances of potentially having that information revealed to peers and thus avoid ensuing stigma. Others may overreport, reporting an aspirational state rather than fact. 

The important issue to track is: these are not random perturbations of the measurement; these are not symmetric; the motivations may be challenging to discern in any particular example but they are likely arising from non-random issues that can be investigated and obviated. Some improvements to the study that reduce the number of overreporters may also reduce the number of underreporters -- for example, using computer-based responses or using randomized response techniques -- because these methods of data collection offer a form of protection against having accidental disclosures lead to social embarrassment. 

Unfortunately, understanding the dynamics of a particular reporting class may have very little benefit on understanding other reporting class dynamics. For example, a method to make many underreporters become truth-tellers may not be useful for getting overreporters to become truth-tellers. But this is also true within a reporting-class: the distribution of reasons for why female overreporters are overreporters may be quite different from the distribution of reasons for why male overreporters overreport. Thus subgroups within a reporting-class may require different modifications to the study protocol to change their reporting-class to truth-teller. This is because there are real-world, causal issues that give rise to ME, and these can differ between, and within, these compliance classes. 

\section{Proof of Theorem \ref{thm:bias}}\label{biasProof}
Observe that we can write: 
\[ \yirep = W_i \bigg( (1-U_i) (1-O_i)  Y_i(1) + O_i \bigg) + (1-W_i) \bigg( (1-U_i) (1-O_i)  Y_i(0) + O_i\bigg)\,. \]
Defining
\[ \tilde Y_i(1) = (1-U_i) (1-O_i)  Y_i(1) + O_i, \hspace{5mm} \tilde Y_i(0) = (1-U_i) (1-O_i)  Y_i(0) + O_i \] 
we then have 
\[ \yirep = W_i \tilde Y_i(1) + (1 - W_i) \tilde Y_i(0) \,\] 
which matches the standard form for the observed outcome.

We begin by proving the first equality. We condition on the sample chosen for the experiment. Within this sample, we observe by analogy that 
\begin{align*}
\E \left( \hat \tau \right) &= \E_R \left( \E_W \left( \frac{1}{n_t} \sum_{i = 1}^N Y_i W_iR_i - \frac{1}{n_c} \sum_{i = 1}^N Y_i(1-W_i)R_i \mid \{R_i\}_{i = 1}^N \right) \right)  \\
&= \E_R \left( \frac{1}{n}  \left( \sum_{i = 1}^N R_i (\tilde Y_i(1) - \tilde Y_i(0)) \right) \right) \\
&= \frac{1}{N} \sum_{i = 1}^N (1 - U_i)(1-O_i)\tau_i \\
&= \frac{1}{N} \sum_{i = 1}^N (1 - U_i - O_i) \tau_i \\
&= \left(1 - \bar U_{\spp} - \bar O_{\spp}\right) \tau_{\spp} - \cov_{\spp}\left(U, \tau \right) - \cov_{\spp}\left(O, \tau \right)\,.
\end{align*}

The second equality follows from the fact that $\tau_i = I_i - D_i$ and thus the second to last line of the above expansion can equivalently be expanded as 
\begin{align*}
 \E \left( \hat \tau \right) &= \frac{1}{N} \sum_{i = 1}^N (1 - U_i - O_i) (I_i - D_i) \\
&= \tau_{sp} - \overline{UI} + \overline{UD} - \overline{OI}  + \overline{OD}\,.
\end{align*}

\section{Proof of Theorem \ref{thm:indepPower}}\label{proof:indepPower}

We can obtain an unbiased estimate for $V$, the true variance of $
\hat \tau$, by computing the Neyman variance estimator, 
\[ \hat V = \frac{s_t^2}{n_t} + \frac{s_c^2}{n_c}\,,  \]
where $s_t^2$ and $s_c^2$ are the empirical variance estimates among treated and control outcomes \citep{imbens2015causal}. Invoking Slutsky's Lemma, we find that 
\begin{equation}\label{eq:limNormal}
    \frac{\hat \tau - \E \left( \hat \tau \right)}{\sqrt{\hat V}} \stackrel{d} \longrightarrow N \left(0, 1 \right) \,.
\end{equation} 

In the testing setting, we compare $\hat \tau/\sqrt{\hat V}$ versus a $N(0, 1)$ reference distribution. If we use a one-sided test with a probability cutoff of $\alpha$, then our critical value is $\Phi^{-1}(\alpha)$ where $\Phi(\cdot)$ is the inverse CDF of a standard normal. The detection probability is 
\begin{align*}
P(\text{detection}) &= P\left(\frac{\hat \tau}{\sqrt{\hat V}} < \Phi^{-1}(\alpha) \right) \\
&= P\left(\frac{\hat \tau - \E \left( \hat \tau \right)}{\sqrt{\hat V}} < \Phi^{-1}(\alpha) - \frac{\E  (\hat\tau)}{\sqrt{\hat V}}\right)\,. \\
\end{align*}
Invoking the distribution limit in (\ref{eq:limNormal}) and the fact that $\sqrt{\hat V} \stackrel{p} \to \sqrt{V}$, we see that
\begin{equation}\label{eq:}
P(\text{detection}) \to \Phi\left( \Phi^{-1}(\alpha) - \frac{\E (\hat \tau)}{\sqrt{V}} \right).
\end{equation}  
This formula gives us a way to compute the relative detection probability in the case of measurement error. 

Under the conditions of the theorem, it follows immediately that 
\[ \cov_{\spp}\left(U, \tau \right) =  \cov_{\spp}\left(O, \tau \right) = 0 \] 
and thus 
\[ \E(\hat \tau) =  \left(1 - \bar U_{\spp} - \bar O_{\spp}\right) \tau_{\spp}\,. \] 
Next, we consider $V$, which is given by 
\[ V = \frac{\tilde \sigma_t^2}{n_t} + \frac{\tilde \sigma_c^2}{n_c}  \] 
where $\tilde \sigma_t^2$ is the population variance of $\tilde Y_i(1)$ and $\tilde \sigma_c^2$ the population variance of $\tilde Y_i(0)$. Because the $\tilde Y_i(1)$ and $\tilde Y_i(0)$ are binary, we see
\begin{align*}
\tilde \sigma_t^2 &= \overline{\tilde Y_i(1)}\cdot (1 - \overline{\tilde Y_i(1)}) 
\end{align*}
where $\overline{\tilde Y_i(1)}$ is the mean of the $\tilde Y_i(1)$ values in the super-population. Under the independence assumption, this simplifies to 
\begin{align*}
\tilde \sigma_t^2 &= \frac{1}{N^2} \left( \sum_i (1-U_i) (1-O_i)  Y_i(1) + O_i \right) \left( \sum_i 1 - (1-U_i) (1-O_i)  Y_i(1) + O_i, \right)\\
&= \left( \left(1 - \bar U_{\spp} - \bar O_{\spp}\right) \overline{Y_i(1)} + \bar O_{\spp} \right)\left( 1 - \left(1 - \bar U_{\spp} - \bar O_{\spp}\right) \overline{Y_i(1)} - \bar O_{\spp} \right) 
\end{align*}
where $\overline{ Y_i(1)}$ is the mean of the $ Y_i(1)$ values in the super-population. The expansion of $\tilde \sigma_c^2$ is analogous.

Putting this together, we see
\begin{align*}
\frac{\E (\hat \tau)}{\sqrt{V}} &= 
   \left( \left(1 - \bar U_{\spp} - \bar O_{\spp}\right) (\overline{Y_i(1)} - \overline{Y_i(0)}\right) \bigg/
    \\ & \left(
    \frac{\left( \left(1 - \bar U_{\spp} - \bar O_{\spp}\right) \overline{Y_i(1)} + \bar O_{\spp} \right)\left( 1 - \left(1 - \bar U_{\spp} - \bar O_{\spp}\right) \overline{Y_i(1)} - \bar O_{\spp} \right)}{n_t}  \right. + \\
    &\left.  \frac{\left( \left(1 - \bar U_{\spp} - \bar O_{\spp}\right) \overline{Y_i(0)} + \bar O_{\spp} \right)\left( 1 - \left(1 - \bar U_{\spp} - \bar O_{\spp}\right) \overline{Y_i(0)} - \bar O_{\spp} \right)}{n_c} \right)^{1/2} \,.
\end{align*}

Taking the gradient of this quantity with respect to $\bar U_{\spp}$ and imposing the relevant domain restrictions, we see that this is a strictly decreasing function of $\bar U_{\spp}$. Since the power is itself an increasing function of $\E (\hat \tau)/\sqrt{V}$, it follows that power decreases as $\bar U_{\spp}$ rises, ceteris paribus. An identical analysis yields the same insight for $\bar O_{\spp}$.

\section{Solving Optimization Problem \ref{mainOptProb}}\label{app:solveOpt}
The method of \cite{dinkelbach1967nonlinear} provides an efficient technique for solving a class of fractional programming problems. In particular, consider the problem 
\begin{equation}\label{optProb:basic}
\min_{x \in \mathcal{X}} \frac{f(x)}{g(x)}
\end{equation}
where $f(x)$ is convex and $g(x)$ is concave and positive for all $x \in \mathcal{X}$. 
Consider also the function
\begin{equation*}
h(\gamma) = \min_{x \in \mathcal{X}} f(x) - \gamma g(x) \,.
\end{equation*}
Dinkelbach showed that $h(\gamma) = 0$ if and only if $\gamma$ is the solution to Optimization Problem \ref{optProb:basic}.

This provides a straightforward framework for solving Problem \ref{app:solveOpt}. In particular observe that 
\[ f(\overline{U}) = \left(\overline{I} - \overline{D} - \overline{U} \left(\overline{I}\delta_I - \overline{D} \delta_D \right)\right) \] 
is convex in $\overline{U}$ and 
\begin{align*}
g(\overline{U}) &= \left( 
\left(\overline{I} + \overline{A} - \overline{U} \left(\overline{I}\delta_I + \overline{A} \delta_A \right)\right)\left(1 - \left(\overline{I} + \overline{A} - \overline{U} \left(\overline{I}\delta_I + \overline{A} \delta_A \right)\right)\right) +\right. \\
&  \left. \left(\overline{D} + \overline{A} - \overline{U} \left(\overline{D}\delta_D + \overline{I} \delta_A \right)\right)\left(1 - \left(\overline{D} + \overline{A} - \overline{U} \left(\overline{D}\delta_D + \overline{A} \delta_A \right)\right)\right) \right) 
\end{align*}
is concave in $\overline{U}$. Hence, the function
\[ h(\lambda) = f(\overline{U}) - \lambda \left(  g(\overline{U}) \right)\] 
is convex in $\overline{U}$ for all positive values of $\lambda$. The optimization problem
\begin{equation}\label{optProb:dink}
\begin{aligned}
\text{minimize} & \hspace{5mm} h(\lambda) \\
\text{subject to} & \hspace{5mm}  \overline{I}\delta_I +  \overline{D}\delta_D + \overline{A}\delta_A + \overline{N}\delta_N = 1,\\ &\hspace{5mm}
\frac{1}{\Gamma} \leq \delta_I, \delta_D, \delta_N, \delta_A \leq \Gamma\,.
\end{aligned}
\end{equation}
is convex and can be solved efficiently. 

This leads to a simple binary search algorithm that can be used to solve Optimization Problem \ref{mainOptProb}. The algorithm inputs are type I error bound $p$, type II error bound $\beta$, and user-supplied estimates of $\overline{I}, \overline{D}, \overline{N}, \overline{A}, \overline{U}$ and $\Gamma$. 

\begin{itemize}
    \item Begin with $\lambda_{lwr} = 0$, a large value $\lambda_{upr}$, and a small value $\epsilon$. 
    \item While $\lambda_{upr} - \lambda_{lwr} > \epsilon$:
    \begin{itemize}
        \item Set $\lambda_{mid} = (\lambda_{upr} + \lambda_{lwr})/2$.
        \item Solve Optimization Problem \ref{optProb:dink} with objective given by $h(\lambda_{mid})$. Denote its optimal value $h^{\star}$.
        \item If $h^{\star} < 0,$ set $\lambda_{upr} \leftarrow \lambda_{mid}$. Otherwise, set $\lambda_{lwr} \leftarrow \lambda_{mid}$. 
    \end{itemize}
    \item Compute the required sample size for each of the treatment and control arms as 
    \[ n =  \frac{\left(\Phi^{-1}(1 - \beta) + \Phi^{-1}(p)\right)^2}{\lambda_{mid}}\,. \] 
\end{itemize}

\section{Power Calculation in Detail}\label{app:powerCalc}

Recall that we have two populations to compare, with sample proportions of students reporting sexual violence in the prior year: $p_T$ and $p_C$. Assume the number of individuals in each treatment group is the same; that is, $n_T = n_C$. We want to calculate the required sample size for a one-sided hypothesis test with power 80\%.  

For a one-sided test, the lower end of our confidence interval will be: 
\[ \hat{\mu} - 1.645 * \hat{se}/\sqrt{n} \] 

To reach power of 80\%, we will need:

\[ \mathbb{P}_{H_a}(\hat{\mu} > 1.645 * \hat{\sigma}/\sqrt{n_T}) \ge 0.20 \]

We assume the intervention reduces sexual violence by 50\% annually, so that $|\mu| = |p_T - p_C| = |(0.07 * 0.50) - 0.07| = 0.035$. So, rearranging to isolate $n$, we get 

\[ n \ge \left( \frac{1.645 + \Phi^{-1}(0.8)}{0.035 / \hat{\sigma}} \right)^2\,. \] 

We can estimate $\hat{\sigma} = \sqrt{p_T(1-p_T) + p_C(1-p_C)} = \sqrt{0.09} = 0.314$. This leads to an estimate of the necessary sample size as: 
\[ n = 2 * n_T = 2 * \{(1.64 + 0.84) * 0.314)/0.035\}^2 = 994 \,.\]

\end{document}